\documentclass[11pt,letterpaper]{article}
\pdfoutput=1

\usepackage[vmargin=1.1in,hmargin=1.2in]{geometry}
\usepackage{amsmath,amssymb,amsthm,booktabs,color,doi,enumitem,graphicx,latexsym,url,xcolor,eucal}
\usepackage[numbers,sort&compress]{natbib}
\usepackage{microtype,hyperref}

\definecolor{citecolor}{HTML}{0000C0}
\definecolor{urlcolor}{HTML}{000080}

\hypersetup{
    colorlinks=true,
    linkcolor=black,
    citecolor=citecolor,
    filecolor=black,
    urlcolor=urlcolor,
}

\newtheorem{theorem}{Theorem}
\newtheorem{lemma}[theorem]{Lemma}

\theoremstyle{remark}

\theoremstyle{definition}
\newtheorem{definition}[theorem]{Definition}

\newcommand{\namedref}[2]{\hyperref[#2]{#1~\ref*{#2}}}

\newcommand{\theoremref}[1]{\namedref{Theorem}{#1}}

\newcommand{\ESS}{\mathcal{S}}
\newcommand{\T}{\mathcal{T}}

\newcommand{\card}[1]{\left\lvert {#1} \right\rvert}
\newcommand{\bigcard}[1]{\bigl\lvert {#1} \bigr\rvert}

\setitemize{noitemsep}
\setenumerate{noitemsep,itemsep=1ex}

\newenvironment{mycover}
               {\list{}{\listparindent 0pt
                        \itemindent    \listparindent
                        \leftmargin    0pt
                        \rightmargin   0pt
                        \parsep        0pt}%
                \raggedright
                \item\relax}
               {\endlist}

\begin{document}

\hypersetup{
    pdfauthor={},
    pdftitle={},
}

\begin{mycover}
{\LARGE \textbf{Deterministic MST Sparsification\\in the Congested Clique}\par}

\bigskip
\bigskip

\textbf{Janne H.\ Korhonen}\\
{\small School of Computer Science, Reykjav\'{\i}k University\par}
\end{mycover}

\bigskip
\paragraph{Abstract.} We give a simple deterministic constant-round algorithm in the congested clique model for reducing the number of edges in a graph to $n^{1+\varepsilon}$ while preserving the minimum spanning forest, where $\varepsilon > 0$ is any constant. This implies that in the congested clique model, it is sufficient to improve MST and other connectivity algorithms on graphs with slightly superlinear number of edges to obtain a general improvement. As a byproduct, we also obtain a simple alternative proof showing that MST can be computed deterministically in $O(\log \log n)$ rounds. 
\bigskip

\clearpage

\section{Introduction}

\paragraph{MST in the congested clique.} The \emph{congested clique}~\cite{lotker05} is a specialisation of the standard CONGEST model of distributed computing; in the congested clique, each of the $n$ nodes of the network can send a different message of $O(\log n)$ bits to each other node each synchronous communication round. The congested clique model has attracted considerable interest recently, as the fully connected communication topology allows for much faster algorithms than the general CONGEST model.

Minimum spanning tree is perhaps the most studied problem in the congested clique model, and a good example of the power of the model. The \citet{lotker05} paper introducing the congested clique model gave an $O(\log \log n)$-round deterministic MST algorithm. Subsequently, even faster randomised algorithms have been discovered: the $O(\log \log \log n)$-round algorithm by \citet{Hegeman15_MST_logloglogn}, and the recent $O(\log^* n)$-round algorithm by \citet{logstarMST}.

\paragraph{MST sparsification.} Both of the above fast randomised MST algorithms are based on fast randomised \emph{graph connectivity} algorithms. To solve MST, they use a reduction of \citet{Hegeman15_MST_logloglogn} from MST to graph connectivity; this works by (1) reducing general MST into two instances of MST on graphs with $O(n^{3/2})$ edges using a randomised sampling technique of \citet{karger1995randomized}, and (2) reducing MST on sparse graphs to multiple independent instances of graph connectivity.

In this work, we take a closer look at the sparsification step of the \citet{Hegeman15_MST_logloglogn} reduction. Specifically, we show that it is possible to do obtain much stronger sparsification for connectivity problems in constant rounds without using randomness:
\begin{theorem}\label{thm}
Given a weighted graph $G = (V,E)$ and an integer $k$, we can compute in $O(k)$ rounds an edge subset $E' \subseteq E$ with $\card{E'} =  O\bigl(n^{1 + 1/2^k}\bigr)$ such that $E'$ contains one minimum spanning forest of $G$.
\end{theorem}

In particular, \theoremref{thm} implies that graphs with slightly superlinear number of edges are the hardest case for connectivity problems in the congested clique, as graphs with linear number of edges can be learned by all nodes in constant rounds using the routing protocol of \citet{lenzen2013optimal}. Alas, our sparsification technique alone fails to improve upon the state-of-the-art even for deterministic MST algorithms, though applying \theoremref{thm} with $k = \log \log n$ does give an alternative deterministic $O(\log \log n)$ algorithm for MST in the congested clique.

\section{Deterministic MST Sparsification}

Let $\ESS = \{ S_1, S_2, \dotsc, S_\ell \}$ be a partition of $V$. For integers $i,j$ with $1 \le i \le j \le \ell$, we define
\[E^\ESS_{ij} = \bigl\{\{ u, v \} \in E \mathbin{:} u \in S_i \text{ and } v \in S_j \bigr\}\,,\]
and denote by $G^\ESS_{ij}$ the subgraph of $G$ with vertex set $S_{i} \cup S_{j}$ and edge set $E^\ESS_{ij}$.

\begin{definition}
For $0 < \varepsilon \le 1$, graph $G = (V, E)$ and partition $\ESS = \{ S_1, S_2, \dotsc, S_\ell \}$ of $V$, we say that $(G,\ESS)$ is \emph{$\varepsilon$-sparse} if $\ell = n^{\varepsilon}$, each $S \in \ESS$ has size at most $n^{1-\varepsilon}$ and for each $i,j$ with $1 \le i \le j \le \ell$, we have $\bigcard{E^{\ESS}_{ij}} \le 2n^{1-\varepsilon}$.
\end{definition}

If $(G,\ESS)$ is $\varepsilon$-sparse, then $G$ can have at most $2n^{1+\varepsilon}$ edges. Moreover, we will now show that we can \emph{amplify} this notion of sparseness from $\varepsilon$ to $\varepsilon/2$ in constant rounds. Observing that for any graph $G = (V,E)$ and $\ESS = \bigl\{ \{ v \} \mathbin{:} v \in V \bigr\}$, we have that $(G,\ESS)$ is $1$-sparse, we can start from arbitrary graph and apply this sparsification $k$ times to obtain sparsity $1/2^k$, yielding \theoremref{thm}.

For convenience, let us assume that the all edge weights in the input graph in distinct, which also implies that the minimum spanning forest is unique. If this is not the case, we can break ties arbitrarily to obtain total ordering of weights. Recall that each node in $V$ receives its incident edges in $G$ as input.

\begin{lemma}\label{lemma1}
Given a graph $G = (V,E)$ with distinct edge weights and unique MSF $F \subseteq E$, and a globally known partition $\ESS$ such that $(G,\ESS)$ is $\varepsilon$-sparse, we can compute a subgraph $G' = (V,E')$ of $G$ and a globally known partition $\T$ such that $(G', \T)$ is $\varepsilon/2$-sparse and $F \subseteq E'$ in constant number of rounds.
\end{lemma}

\begin{proof}
To obtain the partition $\T = \{ T_1, T_2, \dotsc, T_{n^{\varepsilon/2}} \}$, we construct each set $T_i$ by taking the union of $n^{\varepsilon/2}$ sets $S_j$. Clearly sets $T_i$ constructed this way have size $n^{1-\varepsilon}$, and this partition can be constructed by the nodes locally. Since $(G,\ESS)$ is $\varepsilon$-sparse, we now have that
\[ \bigcard{E^\T_{ij}} = \sum_{x \colon S_{x} \subseteq T_i} \sum_{y \colon S_{y} \subseteq T_j} \bigcard{E^{\ESS}_{xy}} \le (n^{\varepsilon/2})^2 2n^{1-\varepsilon} = 2n\,. \]
We assign arbitrarily each pair $(i,j)$ with $1 \le i \le j \le n^{\varepsilon/2}$ as a \emph{label} for distinct node $v \in V$. The number of such pairs $(i,j)$ is at most $(n^{\varepsilon/2})^2 = n^\varepsilon \le n$, so this is always possible, though some nodes may be left without labels. The algorithm now proceeds as follows:
\begin{enumerate}
    \item Distribute information about the edges so that node with label $(i,j)$ knows the full edge set $E^{\T}_{ij}$. Since $\bigcard{E^\T_{ij}} \le 2n$, this can be done in constant rounds using the routing protocol of \citet{lenzen2013optimal}.
    \item Each node with label $(i,j)$ locally computes a minimum spanning forest $F^{\T}_{ij}$ for the subgraph $G^\T_{ij}$ using information obtained in previous step. Since $\card{T_i \cup T_j} \le 2n^{1-\varepsilon/2}$, we also have $\bigcard{F^\T_{ij}} \le 2n^{1-\varepsilon/2}$.
    \item Redistribute information about the sets $F^\T_{ij}$ so that each node knows which of its incident edges are in one of the sets $F^\T_{ij}$. Again, this takes constant rounds.
\end{enumerate}
Taking $E' = \bigcup_{(i,j)} F^\T_{ij}$, we have that $(G',\T)$ is $\varepsilon/2$-sparse.
To see that $E'$ also contains all edges of $F$, recall the fact that an edge $e \in E$ is in MSF $F$ if and only if it is the minimum-weight edge crossing some cut $(V_1, V_2)$, assuming distinct edge weights (see e.g.~\citet{karger1995randomized}). If edge $e \in E^\T_{ij}$ is in $F$, then it is minimum-weight edge crossing a cut $(V_1, V_2)$ in $G$, and thus also minimum-weight edge crossing the corresponding cut in $G^\T_{ij}$, implying $e \in F^\T_{ij}$.
\end{proof}

\bigskip

\paragraph{Acknowledgements.} We thank Magn\'us M. Halld\'orsson, Juho Hirvonen, Tuomo Lempi\"ainen, Christopher Purcell, Joel Rybicki and Jukka Suomela for discussions, and Mohsen Ghaffari for sharing a preprint of \cite{logstarMST}. This work was supported by grant 152679-051 from the Icelandic Research Fund.


\pagebreak

\DeclareUrlCommand{\Doi}{\urlstyle{same}}
\renewcommand{\doi}[1]{\href{http://dx.doi.org/#1}{\footnotesize\sf doi:\Doi{#1}}}
\bibliographystyle{plainnat}
\bibliography{sparsification}

\end{document}